\RequirePackage{fix-cm}
\documentclass[smallextended]{svjour3}       
\smartqed  
\usepackage{graphicx}
\usepackage{amsfonts}
\usepackage{amssymb,epic, graphicx}
\usepackage[cmex10]{amsmath}
\usepackage{array}
\usepackage{color}
\usepackage{mdwmath}
\usepackage{mdwtab}
\usepackage{eqparbox}
\usepackage[tight,footnotesize]{subfigure}
\usepackage{algorithm}
\usepackage{algorithmic}
\allowdisplaybreaks
\begin{document}

\title{The Complete Weight Enumerator of Several Cyclic Codes
}

\titlerunning{The CWE of Several Cyclic Codes}        

\author{Shudi Yang      \and
        Zheng-An Yao
        }


\institute{S.D. Yang \at
              Department of Mathematics,
Sun Yat-sen University, Guangzhou 510275 and School of Mathematical
Sciences, Qufu Normal University, Shandong 273165, P.R.China \\
              Tel.: +86-15602338023\\
                            \email{yangshd3@mail2.sysu.edu.cn}           
           \and
           Z.-A. Yao \at
               Department of Mathematics,
Sun Yat-sen University, Guangzhou 510275, P.R. China}

\date{Received: date / Accepted: date}

\maketitle

\begin{abstract}
 Cyclic codes have attracted a lot of research interest for decades.
 In this paper, for an odd prime $p$, we propose a general strategy to
 compute the complete weight enumerator of cyclic codes via the value
 distribution of the corresponding exponential sums. As applications of this
 general strategy, we determine the complete weight enumerator of several
 $p$-ary cyclic codes and give some examples to illustrate our results.

\keywords{Cyclic code \and Gauss sum \and Exponential sum
\and\\
Weight enumerator \and Complete weight enumerator}
\end{abstract}

\vspace{1\baselineskip}
\noindent$\displaystyle \mathbf{Mathematics~~ Subject~~ Classification~~~} $    11T71$ \cdot$94B15

\section{Introduction}\label{sec:intro}

 Throughout this paper, let $p$ be an odd prime. Denote by $\mathbb{F}_p$ a
finite field with $p$ elements. An $[n, \kappa, \delta\,]$ linear code
$C$ over $\mathbb{F}_p$ is a $\kappa$-dimensional subspace of
$\mathbb{F}_p^n$ with minimum distance $\delta$. Moreover, the code is
cyclic if every codeword $(c_0,c_1,\cdots,c_{n-1})\in C $ whenever
$(c_{n-1},c_0,\cdots,c_{n-2})\in C $. Any cyclic code $C$ of length
$n$ over $\mathbb{F}_p$ can be viewed as an ideal of
$\mathbb{F}_p[x]/(x^n-1)$. Therefore, $C=\left\langle
g(x)\right\rangle $, where $g(x)$ is the monic polynomial of lowest degree
and divides $x^n-1$. Then $g(x)$ is called the
generator polynomial and $h(x)=(x^n-1)/g(x)$ is called the parity-check
polynomial \cite{macwilliams1977theory}.

 The ordinary weight enumerator of $C$ of length $n$ is defined by
$$A_0+A_1x+A_2x^2+\cdots+A_nx^n,$$
where $A_i$ is the number of codewords with Hamming weight
$i$ and $A_0=1$. The sequence $(A_0,A_1,A_2,\cdots,A_n)$ is called the weight
distribution of the code $C$.

The complete weight enumerator of a code $C$ over $\mathbb{F}_p$ enumerates the codewords according
to the number times each element of the field appears in each codeword. Denote the
field elements by $\mathbb{F}_p=\{w_0,w_1,\cdots,w_{p-1}\}$, where $w_0=0$.
Also let $\mathbb{F}_p^*$ denote $\mathbb{F}_p\backslash\{0\}$.
For a codeword $\mathsf{c}=(c_0,c_1,\cdots,c_{n-1})\in \mathbb{F}_p^n$, let $w[\mathsf{c}]$ be the
complete weight enumerator of $\mathsf{c}$ defined as
$$w[\mathsf{c}]=w_0^{k_0}w_1^{k_1}\cdots w_{p-1}^{k_{p-1}},$$
where $k_j$ is the number of components of $\mathsf{c}$ equal to $w_j$, $\sum_{j=0}^{p-1}k_j=n$.
The complete weight enumerator of the code $C$ is then
$$\mathrm{CWE}(C)=\sum_{\mathsf{c}\in C}w[\mathsf{c}].$$

The weight
distribution of a code has been extensively studied for a long time and we refer the reader to \cite{ding2011,ding2013hamming,dinh2015recent,sharma2012weight,vega2012weight,zheng2013weight} and references therein for an overview of the related researches.
Note that the complete weight enumerator of a codeword implies its weight, which indicates
that the weight distribution of the code can be obtained from its complete weight enumerator. The information of the complete weight enumerator of a linear code is of vital use in practical applications. For example, Blake and Kith pointed out that the  complete weight enumerator of Reed-Solomon codes could be helpful  in soft decision decoding~\cite{Blake1991,kith1989complete}. In~\cite{helleseth2006}, the study of the monomial and quadratic bent functions was related to the complete weight enumerators of linear codes. Ding $et~al.$~\cite{ding2007generic,Ding2005auth} showed that the complete weight enumerator can be applied in the computation of the deception probabilities of certain authentication codes. In~\cite{chu2006constant,ding2008optimal,ding2006construction}, the complete weight enumerators of some constant
composition codes were shown to have only one term and some families of optimal constant composition codes were presented.

However, only a few works were focused on the determination of the complete
weight enumerator of linear codes in the literature besides the above mentioned~\cite{Blake1991,kith1989complete,chu2006constant,ding2008optimal,ding2006construction}.
The complete weight enumerators of the
generalized Kerdock code and related linear codes over Galois rings were determined by Kuzmin and Nechaev in~\cite{kuzmin1999complete,kuzmin2001complete}. The authors obtained the complete weight enumerators of some cyclic codes by using exponential sums in~\cite{BaeLi2015complete,li2015complete}. In this paper, we shall determine the complete weight enumerators of a class of cyclic codes related to some special quadratic forms.

Let $m$ and $l$ be two positive integers with $m>l$. For now on, we denote by $\alpha$
a primitive
element of $\mathbb{F}_{p^m}$. Let
$h_1(x)$ and $h_2(x)$ be the minimal polynomials of $\alpha^{-(p^l+1)}$
and $\alpha^{-2}$ over $\mathbb{F}_p$, respectively. Obviously,
$h_1(x)$ and $h_2(x)$ are distinct and
$\mathrm{deg}(h_2(x))=m$. Moreover, it can be easily shown that
$\mathrm{deg}(h_1(x))=m/2$ if $m=2l$ and $m$ otherwise.

Let $C_1$ and $C_2$ be two cyclic codes over $\mathbb{F}_p$ of length
$p^m-1$ with parity-check polynomials $h_1(x)$and
$h_1(x)h_2(x)$, respectively. Hence, for the dimensions of $C_1$ and $C_2$, we have
\begin{eqnarray*}\mathrm{dim}_{\mathbb{F}_{p}}C_1=\left\{\begin{array}{lll}\frac{1}{2}m,&&~~\mathrm{if}~~m=2l,\\
m,&&~~ \mathrm{otherwise},\end{array}
\right.
\end{eqnarray*}and
\begin{eqnarray*}\mathrm{dim}_{\mathbb{F}_{p}}C_2=\left\{\begin{array}{lll}\frac{3}{2}m,&&~~\mathrm{if}~~m=2l,\\
2m,&&~~ \mathrm{otherwise}.\end{array}
\right.
\end{eqnarray*}

From the well-known Delsarte¡¯s Theorem \cite{delsarte1975subfield}, we have the trace representation of $C_1$ and $C_2$ described by
\begin{eqnarray*}
    C_1&=&\{(\mathrm{Tr}^m_1(ax^{p^l+1}))_{x\in\mathbb{F}_{p^m}^*}:a\in\mathbb{F}_{p^m}\},\\
    C_2&=&\{(\mathrm{Tr}^m_1(ax^{p^l+1}+bx^2))_{x\in\mathbb{F}_{p^m}^*}:a,b\in\mathbb{F}_{p^m}\}.
\end{eqnarray*}

The weight distribution of $C_1$ is trivial and can be easily
obtained since the value distributions of the corresponding exponential sums are already known(see \cite{coulter1998explicit,draper2007explicit}). However, to the best of our knowledge, there are no information about its complete weight enumerator. The cyclic code $C_2$ was investigated in the literature. Luo and Feng~\cite{luo2008weight} studied its weight distribution explicitly. Bae, Li and Yue~\cite{BaeLi2015complete} established its complete weight enumerator in the special case of $\mathrm{gcd}(m,l)=1$. In this paper, we will explicitly present the complete weight enumerators of $C_1$ and $C_2$ in view of the relationship between $\upsilon_2(m)$ and $\upsilon_2(l)$ for arbitrary $m$ and $l$ with $m>l$, where $\upsilon_2(\cdot)$ is the 2-adic order function. Thus, we will extend the results in~\cite{BaeLi2015complete} to some extent.

The aim of this paper is to investigate the complete weight enumerators for cyclic codes by utilizing the
theories of Gauss sums and exponential sums over finite fields. A general strategy is proposed
and then applied to determine the complete weight
enumerators for the codes $C_1$ and $C_2$, respectively.

The remainder of this paper is organized as follows. In Section
\ref{sec:Preli}, we introduce some definitions and auxiliary results on
quadratic forms, Gauss sums and exponential sums. Section \ref{sec:main} gives the
main results of this paper, including a general strategy for cyclic codes and the explicitly complete weight
enumerators for the codes $C_1$ and $C_2$.
Section \ref{sec:conclusion} concludes this paper and makes some
remarks on this topic.

\section{Preliminaries}\label{sec:Preli}

We follow the notations in Section \ref{sec:intro}. Let $q$ be a
power of $p$ and $t$ be a positive integer. By identifying
the finite field $\mathbb{F}_{q^t}$ with a $t$-dimensional vector
space $\mathbb{F}^t_{q}$ over $\mathbb{F}_{q}$, a function $f(x)$
from $\mathbb{F}_{q^t}$ to $\mathbb{F}_{q}$ can be regarded as a
$t$-variable polynomial over $\mathbb{F}_{q}$. The function $f(x)$ is called
a quadratic form if it can be written as a
homogeneous polynomial of degree two on  $\mathbb{F}^t_{q}$ as
follows:
$$f(x_1,x_2,\cdots,x_t)=\sum_{1\leqslant i \leqslant j\leqslant t}a_{ij}x_ix_j,~~a_{ij}\in \mathbb{F}_{q}.$$
Here we fix a basis of  $\mathbb{F}^t_{q}$ over  $\mathbb{F}_{q}$
and identify each $x\in \mathbb{F}_{q^t}$ with a vector
$(x_1,x_2,\cdots,x_t)\in\mathbb{F}^t_{q}$.
 The rank of the quadratic form $f(x)$, rank$(f)$, is defined as the
codimension of the $\mathbb{F}_{q}$-vector space
$$W=\{x\in \mathbb{F}_{q^t}|f(x+z)-f(x)-f(z)=0, ~~\mathrm{for~~all}~~z\in \mathbb{F}_{q^t}\}.$$
Then $|W|=q^{t-\mathrm{rank}(f)}$.

For a quadratic form $f(x)$ with $t$ variables over $\mathbb{F}_q$,
there exists a symmetric matrix $A$ over $\mathbb{F}_q$
such that $f(x)=XAX'$, where $X=(x_1,x_2,\cdots,x_t)\in
\mathbb{F}^t_q$ and $X'$ denotes the transpose of $X$. It is known
that there exists a nonsingular matrix $B$ over $\mathbb{F}_q$ such
that $BAB'$ is a diagonal matrix. Making a nonsingular linear
substitution $X=YB$ with $Y=(y_1,y_2,\cdots,y_t)\in \mathbb{F}^t_q$,
we have
$$f(x)=Y(BAB')Y'=\sum^r_{i=1}a_iy^2_i,\,\,\,a_i\in \mathbb{F}^*_q,$$
where $r$ is the rank of $f(x)$. The determinant $\mathrm{det}(f)$
of $f(x)$ is defined to be the determinant of $A$, and $f(x)$ is said to be
nondegenerate if $\mathrm{det}(f)\neq0$.

The quadratic character over $\mathbb{F}_{p^m}$ is defined by
\begin{eqnarray*}
\eta(x)=\left\{\begin{array}{lll}1,
&&\mathrm{if ~~}x \mathrm{~~is~~ a~~ square~~ in~~} \mathbb{F}_{p^m}^*,\\
-1,&&\mathrm{if ~~}x \mathrm{~~is~~ a~~ nonsquare~~ in~~} \mathbb{F}_{p^m}^*,\\
0,&&\mathrm{if ~~}x=0.\\
\end{array}
\right.
\end{eqnarray*}

The canonical additive character of $\mathbb{F}_{p^m}$, denoted $\chi$, is given by
\begin{eqnarray*}
\chi(x)=\zeta_p^{\mathrm{Tr}^m_1(x)}
\end{eqnarray*}for all $x\in \mathbb{F}_{p^m}$, where $\zeta_p=e^{2\pi\sqrt{-1}/p}$ and $\mathrm{Tr}^m_1$ is a trace function from
$\mathbb{F}_{p^m}$ to $\mathbb{F}_{p}$ defined by
$$\mathrm{Tr}^m_1(x)=\sum^{m-1}_{i=0}x^{p^i},~~x\in
\mathbb{F}_{p^m}.$$

To this end, we shall introduce the Gauss sum $G(\eta,\chi)$ over $\mathbb{F}_{p^m}$ given by
\begin{eqnarray*}
G(\eta,\chi)=\sum_{x\in\mathbb{F}_{p^m}^*}\eta(x)\chi(x)=\sum_{x\in\mathbb{F}_{p^m}}\eta(x)\chi(x),
\end{eqnarray*}
and the Gauss sum $G(\bar{\eta},\bar{\chi})$ over $\mathbb{F}_{p}$ given by
\begin{eqnarray*}
G(\bar{\eta},\bar{\chi})=\sum_{x\in\mathbb{F}_{p}^*}\bar{\eta}(x)\bar{\chi}(x)=\sum_{x\in\mathbb{F}_{p}}\bar{\eta}(x)\bar{\chi}(x),
\end{eqnarray*}
where $\bar{\eta}$ and $\bar{\chi}$ are the quadratic and canonical additive characters of $\mathbb{F}_{p}$, respectively.

The lemmas presented below will turn out to be of use in the
sequel.

\begin{lemma}(See Theorems 5.15 \cite{lidl1983finite})\label{lm:gauss sum}
With the symbols and notation above, we have
\begin{eqnarray*}\label{eq:Gausspm}
G(\eta,\chi)=(-1)^{m-1}(\sqrt{-1})^{\frac{(p-1)^2}{4}m}p^{\frac{m}{2}},
\end{eqnarray*}
and
\begin{eqnarray*}\label{eq:Gaussp}
G(\bar{\eta},\bar{\chi})=(\sqrt{-1})^{\frac{(p-1)^2}{4}}p^{\frac{1}{2}}.
\end{eqnarray*}

\end{lemma}

\begin{lemma}(See Theorem 5.33 of \cite{lidl1983finite})\label{lm:expo sum}
With the symbols and notation above. Let $f(x)=a_2x^2+a_1x+a_0\in \mathbb{F}_{p^m}[x]$ with
$a_2\neq0$. Then
\begin{eqnarray*}\label{eq:expo sum}
\sum_{x\in
\mathbb{F}_{p^m}}\chi(f(x))=\chi(a_0-a_1^2(4a_2)^{-1})\eta(a_2)G(\eta,\chi).
\end{eqnarray*}
\end{lemma}

Let $d=\mathrm{gcd}(m,l)$ denote the greatest common
divisor of $m$ and $l$. Take $s={m}/{d}$. In the sequel we will require the following lemma whose
proof can be found in \cite{coulter1998explicit,draper2007explicit,yu2014weight}.

\begin{lemma}\label{lm: exponentialsums}
Let $S(a)=\sum_{x\in
\mathbb{F}_{p^m}}\zeta_p^{\mathrm{Tr}(ax^{p^l+1})}$ and $d=\mathrm{gcd}(m,l)$.
Let $\upsilon_2(\cdot)$ denote
the 2-adic order function.
Then $Q(x)=\mathrm{Tr}(ax^{p^l+1})$ is a quadratic form and for any
$a\in \mathbb{F}^*_{p^m}$,\\
 \textcircled{1} If
$\upsilon_2(m)\leqslant \upsilon_2(k)$, then
 $\mathrm{rank}(Q(x))=m$ and
\begin{eqnarray}\label{eq:S(a)12}
S(a)=\left\{\begin{array}{lll}~~\sqrt{(-1)^{\frac{p^d-1}{2}}}~p^{\frac{m}{2}},
&&\frac{p^m-1}{2}~~times,\\
-\sqrt{(-1)^{\frac{p^d-1}{2}}}~p^{\frac{m}{2}},
&&\frac{p^m-1}{2}~~times.\\
\end{array}
\right.
\end{eqnarray}
\textcircled{2} If $\upsilon_2(m)=\upsilon_2(k)+1$, then
$\mathrm{rank}(Q(x))=m$ or $m-2d$ and
\begin{eqnarray}\label{eq:S(a)3}
S(a)=\left\{\begin{array}{lll}-p^{\frac{m}{2}},
&&~~~\frac{p^d(p^m-1)}{p^d+1}~~times,\\
~~p^{\frac{m}{2}+d}, &&~~~\frac{p^m-1}{p^d+1}~~~~~~times.\\
\end{array}
\right.
\end{eqnarray}
\textcircled{3} If $\upsilon_2(m)>\upsilon_2(k)+1$, then
$\mathrm{rank}(Q(x))=m$ or $m-2d$ and
\begin{eqnarray}\label{eq:S(a)4}
S(a)=\left\{\begin{array}{lll}~~p^{\frac{m}{2}},
&&~~\frac{p^d(p^m-1)}{p^d+1}~~times,\\
-p^{\frac{m}{2}+d}, &&~~\frac{p^m-1}{p^d+1}~~~~~~times.\\
\end{array}
\right.
\end{eqnarray}
\end{lemma}

The following lemma gives the value distribution of the exponential sum
\begin{equation*}
T(a,b)=\sum_{x\in \mathbb{F}_{p^m}}\zeta_p^{\mathrm{Tr}^m_1(ax^{p^l+1}+bx^2)}.
\end{equation*}

\begin{lemma}(See Lemma 2 and Theorem 1 of \cite{luo2008weight})\label{lem:exposumSoddeven}
When $(a,b)$ runs through $\mathbb{F}^2_{p^m}\backslash\{(0,0)\}$, the quadratic form
$\mathrm{Tr}^m_1(ax^{p^l+1}+bx^2)$ has possible rank $m$, $m-d$ or $m-2d$, and

(i) For $s$ being odd, the exponential sum $T(a,b)$ has the following value distribution:
\begin{equation*}
\left\{\begin{array}{lll} \sqrt{(-1)^{\frac{p^d-1}{2}}}p^{\frac m2},&&|R_1|~~times,\\
-\sqrt{(-1)^{\frac{p^d-1}{2}}}p^{\frac m2},&&|R_1|~~times,\\
p^{\frac{m+d}2},&&|R_2|~~times,\\
-p^{\frac{m+d}2},&&|R_3|~~times,\\
\sqrt{(-1)^{\frac{p^d-1}{2}}}p^{\frac {m+2d}2},&&|R_4|~~times,\\
-\sqrt{(-1)^{\frac{p^d-1}{2}}}p^{\frac {m+2d}2},&&|R_4|~~times,\\
\end{array}
\right.
\end{equation*} where $|R_i|$ is given by
\begin{eqnarray}\label{def:Ri}
\left\{\begin{array}{lll}
|R_1|&=&\frac{(p^{m+2d}-p^{m+d}-p^m+p^{2d})(p^m-1)}{2(p^{2d}-1)},\\
|R_2|&=&\frac12(p^{m-d}+p^{\frac{m-d}2})(p^m-1),\\
|R_3|&=&\frac12(p^{m-d}-p^{\frac{m-d}2})(p^m-1),\\
|R_4|&=&\frac{(p^{m-d}-1)(p^m-1)}{2(p^{2d}-1)}.
\end{array}
\right.
\end{eqnarray}

(ii)For $s$ being even, the exponential sum $T(a,b)$ has the following value distribution:
\begin{equation*}
\left\{\begin{array}{lll}
 p^{\frac m2}, && |K_1|~~times, \\
 -p^{\frac m2},&& |K_2| ~~times, \\
 \sqrt{(-1)^{\frac{p^d-1}{2}}}p^{\frac {m+d}2}, & & |K_3| ~~times,\\
  -\sqrt{(-1)^{\frac{p^d-1}{2}}}p^{\frac {m+d}2}, & &  |K_3|~~times,\\
  p^{\frac m2+d}, && |K_4|~~times, \\
   -p^{\frac m2+d} && |K_5| ~~times.\\
\end{array}
\right.
\end{equation*}
where $|K_i|$ is given by
\begin{eqnarray}\label{def:Ki}
\left\{\begin{array}{lll}
|K_1|&=&\frac{(p^{m+2d}-p^{m+d}-p^{m}+p^{\frac m2+2d}-p^{\frac m2+d}+p^{2d})(p^{m}-1)}{2(p^{2d}-1)} ,\\
|K_2|&=&\frac{(p^{m+2d}-p^{m+d}-p^{m}-p^{\frac m2+2d}+p^{\frac m2+d}+p^{2d})(p^{m}-1)}{2(p^{2d}-1)},\\
|K_3|&=&\frac12p^{m-d}(p^{m}-1),\\
|K_4|&=&\frac{(p^{\frac m2}-1)(p^{\frac m2-d}+1)(p^m-1)}{2(p^{2d}-1)},\\
|K_5|&=&\frac{(p^{\frac m2}+1)(p^{\frac m2-d}-1)(p^m-1)}{2(p^{2d}-1) }.
\end{array}
\right.
\end{eqnarray}
\end{lemma}

\section{Main results}\label{sec:main}
This section investigates the complete weight enumerators of cyclic codes by utilizing the value distributions of
the corresponding exponential sums. A general strategy is given and then used to special codes $C_1$, $C_1$ and $C_2$, respectively, as
depicted in Section \ref{sec:intro}.
\subsection{The General Strategy }

We set up our strategy for the general situation which will be used throughout this paper.

Let
\begin{equation*}
    f_{a_{00},\cdots,a_{kk}}(x)=\sum_{i,j=0}^k a_{ij}x^{p^{i}+p^{j}}
\end{equation*}be a polynomial over $\mathbb{F}_{p^m}$, where $k\leqslant m-1$.
It can be verified that $\mathrm{Tr}^m_1(f_{a_{00},\cdots,a_{kk}}(x))$ is a quadratic form over $\mathbb{F}_{p^m}$. The rank of $\mathrm{Tr}^m_1(f_{a_{00},\cdots,a_{kk}}(x))$ is denoted by $r$.

Consider the exponential sum
\begin{equation*}
    S(a_{00},\cdots,a_{kk})=\sum_{x\in\mathbb{F}_{p^m}}\zeta_p^{\mathrm{Tr}^m_1(f_{a_{00},\cdots,a_{kk}}(x))}.
\end{equation*}

We suppose that the sum $S(a_{00},\cdots,a_{kk})$ has been completely determined by the quadratic form $\mathrm{Tr}^m_1(f_{a_{00},\cdots,a_{kk}}(x))$ over
$\mathbb{F}_{p^m}$. In addition, let $f_{r,\beta}$ denote the frequency of $S(a_{00},\cdots,a_{kk})$ taking the value $S_{r,\beta}$ with rank $r$, for $\beta\in J$, where $J$ is an index set.

Now we focus on the complete weight enumerator of the code
\begin{equation}\label{def:genecode}
    C=\{(\mathrm{Tr}^m_1(f_{a_{00},\cdots,a_{kk}}(x)))_{x\in\mathbb{F}_{p^m}^*}:
    a_{00},\cdots,a_{kk}\in\mathbb{F}_{p^m}\}.
\end{equation}

If $a_{00}=\cdots=a_{kk}=0$, the corresponding codeword is the zero codeword, and the contribution to the complete
weight enumerator is
$$w_0^{p^m-1}.$$
Now consider the case that some $a_{ij}$ is nonzero for $1\leqslant i,j\leqslant k$. Let $n_{a_{00},\cdots,a_{kk}}(\rho)$ denote the number of solutions $x\in\mathbb{F}_{p^m}^*$ such that $\mathrm{Tr}^m_1(f_{a_{00},\cdots,a_{kk}}(x))=\rho$, where $\rho\in\mathbb{F}_{p}$, i.e.,
\begin{equation*}
    n_{a_{00},\cdots,a_{kk}}(\rho)=\sharp\{x\in\mathbb{F}_{p^m}^*:\mathrm{Tr}^m_1(f_{a_{00},\cdots,a_{kk}}(x))=\rho\}.
\end{equation*}Then, the contributions of such terms to the complete
weight enumerator are of the form
$$\prod_{\rho=0}^{p-1}w_\rho^{n_{a_{00},\cdots,a_{kk}}(\rho)},$$
and we only need to compute the frequency of each such term and the value of $n_{a_{00},\cdots,a_{kk}}(\rho)$ which will yield the
complete weight enumerator of the code.

Consider the number of solutions $x\in\mathbb{F}_{p^m}$ such that $\mathrm{Tr}^m_1(f_{a_{00},\cdots,a_{kk}}(x))=\rho$, which is given by
\begin{equation*}
    N_{a_{00},\cdots,a_{kk}}(\rho)=\sharp\{x\in\mathbb{F}_{p^m}:\mathrm{Tr}^m_1(f_{a_{00},\cdots,a_{kk}}(x))=\rho\}.
\end{equation*}

It is straightforward that
\begin{equation}\label{genl:Nandn}
    n_{a_{00},\cdots,a_{kk}}(\rho)=\left\{\begin{array}{lll}N_{a_{00},\cdots,a_{kk}}(\rho)-1,&&~~\mathrm{if}~~\rho=0,\\
N_{a_{00},\cdots,a_{kk}}(\rho),&&~~ \mathrm{otherwise}.\end{array}
\right.
\end{equation}

Therefore, it suffices to study the value of $N_{a_{00},\cdots,a_{kk}}(\rho)$, which is determined by
\begin{eqnarray}\label{genl:eq}
    N_{a_{00},\cdots,a_{kk}}(\rho)&=&
        \frac{1}{p}\sum_{x\in\mathbb{F}_{p^m}}\sum_{y\in\mathbb{F}_{p}}\zeta_p^{y(\mathrm{Tr}^m_1(f_{a_{00},\cdots,a_{kk}}(x))-\rho)}\nonumber\\
    &=& p^{m-1}+\frac{1}{p}\sum_{y\in\mathbb{F}_{p}^*}\zeta_p^{y\rho}\sum_{x\in\mathbb{F}_{p^m}}\zeta_p^{y\mathrm{Tr}^m_1(f_{a_{00},\cdots,a_{kk}}(x))}\nonumber\\
    &=& p^{m-1}+\frac{1}{p}\sum_{y\in\mathbb{F}_{p}^*}\zeta_p^{y\rho}S(y a_{00},\cdots,y a_{kk})\nonumber\\
    &=& p^{m-1}+\frac{1}{p}S(a_{00},\cdots, a_{kk})\sum_{y\in\mathbb{F}_{p}^*}\zeta_p^{y\rho}\bar{\eta}(y^r),
\end{eqnarray}
where the last equal sign holds since
\begin{equation*}
    S(y a_{00},\cdots,y a_{kk})=\bar{\eta}(y^r)S(a_{00},\cdots,a_{kk})
\end{equation*} and $\bar{\eta}$ is the quadratic character over $\mathbb{F}_{p}$.

If $\rho=0$, Equation \eqref{genl:eq} shows that
\begin{eqnarray}\label{gene:eq N0}
    N_{a_{00},\cdots,a_{kk}}(0)&=&p^{m-1}+\frac{1}{p}S(a_{00},\cdots,a_{kk})\sum_{y\in\mathbb{F}_{p}^*}\bar{\eta}(y^r)\nonumber\\
    &=&\left\{\begin{array}{lll}p^{m-1}+\frac{p-1}{p}S(a_{00},\cdots,a_{kk}),&&~~\mathrm{if~~}r \mathrm{~~even},\\
p^{m-1},&&~~\mathrm{if~~}r \mathrm{~~odd}.\\\end{array}
\right.
\end{eqnarray} and consequently
\begin{eqnarray}\label{gene:eq n0}
    n_{a_{00},\cdots,a_{kk}}(0)
    =\left\{\begin{array}{lll}p^{m-1}+\frac{p-1}{p}S(a_{00},\cdots,a_{kk})-1,&&~~\mathrm{if~~}r \mathrm{~~even},\\
p^{m-1}-1,&&~~\mathrm{if~~}r \mathrm{~~odd}.\\\end{array}
\right.
\end{eqnarray}

If $\rho\in\mathbb{F}_{p}^*$, it follows from Equations \eqref{genl:Nandn} and \eqref{genl:eq}
that
\begin{eqnarray}\label{genl:eq ne0}
    n_{a_{00},\cdots,a_{kk}}(\rho)&=&N_{a_{00},\cdots,a_{kk}}(\rho)\nonumber\\
    &=&\left\{\begin{array}{lll}p^{m-1}-\frac{1}{p}S(a_{00},\cdots,a_{kk}),&&~~\mathrm{if~~}r \mathrm{~~even},\\
p^{m-1}+\frac{1}{p}\bar{\eta}(\rho)S(a_{00},\cdots,a_{kk})G(\bar{\eta},\bar{\chi}),&&~~\mathrm{if~~}r \mathrm{~~odd}.\\\end{array}
\right.
\end{eqnarray}

By assumption that $f_{r,\beta}$ to be the frequency of $S(a_{00},\cdots,a_{kk})$ taking the value $S_{r,\beta}$ with rank $r$,
each term $\prod_{\rho=0}^{p-1}w_\rho^{n_{a_{00},\cdots,a_{kk}}(\rho)}$ appears $f_{r,\beta}$ times according to
the value of $S(a_{00},\cdots,a_{kk})$ with rank $r$. Clearly, $n_{a_{00},\cdots,a_{kk}}(\rho)$
is related to $S_{r,\beta}$ and thus we denote it by $ n_{a_{00},\cdots,a_{kk}}(\rho;S_{r,\beta})$ to show this.
Therefore, the complete weight enumerator for the code $C$ is
\begin{equation*}
    \mathrm{CWE}(C)=w_0^{p^m-1}+\sum_{r,\beta}f_{r,\beta}\prod_{\rho=0}^{p-1}
    w_\rho^{n_{a_{00},\cdots,a_{kk}}(\rho;S_{r,\beta})}.
\end{equation*}

\subsection{The complete weight enumerator of the code $C_1$ }
Recall that
\begin{equation*}
    C_1=\{\mathsf{c}_2(a)=(\mathrm{Tr}^m_1(ax^{p^l+1}))_{x\in\mathbb{F}_{p^m}^*}:a\in\mathbb{F}_{p^m}\},
\end{equation*}which is a special case of \eqref{def:genecode}.

Now we deal with the complete weight enumerator of the code $ C_1$ by using the exponential sum
\begin{equation*}
    S(a)=\sum_{x\in\mathbb{F}_{p^m}}\zeta_p^{\mathrm{Tr}^m_1(ax^{p^l+1})}.
\end{equation*}

\begin{theorem}\label{thm:code 2} With notation given
before.\\
(i) Assume that $m\neq2l$. Then $C_1$ is a $[p^m-1, m]$ cyclic code over $\mathbb{F}_{p}$ and its complete weight enumerator is shown as follows:\\
\textcircled{1} If
$0=\upsilon_2(m)\leqslant \upsilon_2(l)$, then
\begin{eqnarray}\label{cwe:C2case1}
    \mathrm{CWE}(C_1)&=&w_0^{p^m-1}+\frac{p^m-1}{2}w_0^{p^{m-1}-1}
    \prod_{\rho\in\mathbb{F}_{p}^*}w_{\rho}^{p^{m-1}+\bar{\eta}(\rho)p^{\frac{m-1}{2}}}\nonumber\\
       && +\frac{p^m-1}{2}w_0^{p^{m-1}-1}
       \prod_{\rho\in\mathbb{F}_{p}^*}w_{\rho}^{p^{m-1}-\bar{\eta}(\rho)p^{\frac{m-1}{2}}}.
\end{eqnarray}
\textcircled{2} If
$1\leqslant \upsilon_2(m)\leqslant \upsilon_2(l)$, then
\begin{eqnarray}\label{cwe:C2case2}
    \mathrm{CWE}(C_1)&=&w_0^{p^m-1}+\frac{p^m-1}{2}w_0^{p^{m-1}-1+(p-1)p^{\frac{m-2}{2}}}
    \prod_{\rho\in\mathbb{F}_{p}^*}w_{\rho}^{p^{m-1}-p^{\frac{m-2}{2}}}\nonumber\\
       && +\frac{p^m-1}{2}w_0^{p^{m-1}-1-(p-1)p^{\frac{m-2}{2}}}
       \prod_{\rho\in\mathbb{F}_{p}^*}w_{\rho}^{p^{m-1}+p^{\frac{m-2}{2}}}.
\end{eqnarray}
\textcircled{3} If
$ \upsilon_2(m)= \upsilon_2(l)+1$, then
\begin{eqnarray}\label{cwe:C2case3}
    \mathrm{CWE}(C_1)&=&w_0^{p^m-1}+\frac{p^d(p^m-1)}{p^d+1}w_0^{p^{m-1}-1-(p-1)p^{\frac{m-2}{2}}}
    \prod_{\rho\in\mathbb{F}_{p}^*}w_{\rho}^{p^{m-1}+p^{\frac{m-2}{2}}}\nonumber\\
       && +\frac{p^m-1}{p^d+1}w_0^{p^{m-1}-1+(p-1)p^{\frac{m+2d-2}{2}}}
       \prod_{\rho\in\mathbb{F}_{p}^*}w_{\rho}^{p^{m-1}-p^{\frac{m+2d-2}{2}}}.
\end{eqnarray}
\textcircled{4} If
$ \upsilon_2(m)> \upsilon_2(l)+1$, then
\begin{eqnarray}\label{cwe:C2case4}
    \mathrm{CWE}(C_1)&=&w_0^{p^m-1}\!+\!\frac{p^d(p^m-1)}{p^d+1}w_0^{p^{m-1}-1+(p-1)p^{\frac{m-2}{2}}}
    \prod_{\rho\in\mathbb{F}_{p}^*}w_{\rho}^{p^{m-1}-p^{\frac{m-2}{2}}}\nonumber\\
       && \!+\frac{p^m-1}{p^d+1}w_0^{p^{m-1}-1-(p-1)p^{\frac{m+2d-2}{2}}}
       \prod_{\rho\in\mathbb{F}_{p}^*}w_{\rho}^{p^{m-1}+p^{\frac{m+2d-2}{2}}}.
\end{eqnarray}
(ii) Assume that $m=2l$. Then $C_1$ is a $[p^m-1, m/2]$ cyclic code over $\mathbb{F}_{p}$ and its complete weight enumerator is given by
\begin{eqnarray}\label{cwe:C2case32}
    \mathrm{CWE}(C_1)=w_0^{p^m-1}\!+\!(p^{\frac{m}{2}}-1)w_0^{p^{m-1}-1-(p-1)p^{\frac{m-2}{2}}}
    \prod_{\rho\in\mathbb{F}_{p}^*}w_{\rho}^{p^{m-1}+p^{\frac{m-2}{2}}}.
       \end{eqnarray}
\end{theorem}
\begin{proof}
(i) Assume that $m\neq2l$. We only give the proof for the case
 $0=\upsilon_2(m)\leqslant \upsilon_2(l)$ since other cases are similar.

Clearly $a=0$ gives the zero codeword and the contribution to the complete
weight enumerator is
$$w_0^{p^m-1}.$$

Consider $a\in\mathbb{F}_{p^m}^*$. Let
\begin{equation*}
    N_{a}(\rho)=\sharp\{x\in\mathbb{F}_{p^m}:\mathrm{Tr}^m_1(ax^{p^l+1})=\rho\}.
\end{equation*}
and \begin{equation*}
    n_{a}(\rho)=\sharp\{x\in\mathbb{F}_{p^m}^*:\mathrm{Tr}^m_1(ax^{p^l+1})=\rho\}.
\end{equation*}

Note that $r=m$ and $m$ is odd. By Equations \eqref{gene:eq N0} and \eqref{genl:eq ne0}, we have
$$ N_{a}(0)=p^{m-1},$$
and for a fixed $\rho\in\mathbb{F}_{p}^*$,
 \begin{eqnarray*}
 N_{a}(\rho)&=& p^{m-1}+\frac{1}{p}\bar{\eta}(\rho)S(a)G(\bar{\eta},\bar{\chi}).
 \end{eqnarray*}

 It then follows from Equation \eqref{eq:S(a)12} and Lemma \ref{eq:Gausspm} that
  \begin{eqnarray*}
 N_{a}(\rho)=\left\{\begin{array}{lll}
 p^{m-1}+\bar{\eta}(\rho)\sqrt{-1}^{\frac{(p-1)^2}{4}+\frac{p^d-1}{2}}p^{\frac{m-1}{2}},&&~~\frac{p^m-1}{2}~~\mathrm{times},\\
 p^{m-1}-\bar{\eta}(\rho)\sqrt{-1}^{\frac{(p-1)^2}{4}+\frac{p^d-1}{2}}p^{\frac{m-1}{2}},&&~~\frac{p^m-1}{2}~~\mathrm{times}.\end{array}
\right.
 \end{eqnarray*}

 Note that $\frac{(p-1)^2}{4}+\frac{p^d-1}{2}$ is an even integer. This implies that
 \begin{eqnarray*}
 N_{a}(\rho)=\left\{\begin{array}{lll}
 p^{m-1}+\bar{\eta}(\rho)p^{\frac{m-1}{2}},&&~~\frac{p^m-1}{2}~~\mathrm{times},\\
 p^{m-1}-\bar{\eta}(\rho)p^{\frac{m-1}{2}},&&~~\frac{p^m-1}{2}~~\mathrm{times}.\end{array}
\right.
 \end{eqnarray*}

 By Equation \eqref{genl:Nandn} and the above analysis, the result given by Equation \eqref{cwe:C2case1} holds for the case
 $0=\upsilon_2(m)\leqslant \upsilon_2(l)$.

(ii) Assume that $ m=2l$.

Let $$K=\{x\in\mathbb{F}_{p^m}\big|~x^{p^l}+x=0\}.$$
 Note that $\mathsf{c}_2(a)=\mathsf{c}_2(a+\tau)$
 for any $\tau\in K$ and $\mathsf{c}_2(a)\in C_1$. Hence, $C_1$ is degenerate with
 dimension $m/2$ over $\mathbb{F}_p$.

Clearly $|K|=p^{\frac{m}{2}}$ and $\upsilon_2(m)=\upsilon_2(l)+1$. Substituting $d={m}/{2}$ to
 Equation \eqref{cwe:C2case3} and dividing each $A_i$ by
 $p^{\frac{m}{2}}$, we get the result given by \eqref{cwe:C2case32}.

This finishes the proof of Theorem \ref{thm:code 2}.
\hfill\space$\qed$\end{proof}

\begin{example}
(i) Let $m=3$, $k=1$, $p=5$. This corresponds to the case $0=\upsilon_2(m)\leqslant \upsilon_2(l)$. Magma works out that the complete weight enumerator for the code $C_1$ is
$$ w_0^{124} + 62w_0^{24}w_1^{30}w_2^{20}w_3^{20}w_4^{30} +
    62w_0^{24}w_1^{20}w_2^{30}w_3^{30}w_4^{20}.$$

(ii) Let $m=6$, $k=2$, $p=3$. This corresponds to the case $1\leqslant\upsilon_2(m)\leqslant \upsilon_2(l)$. Magma shows that the complete weight enumerator for the code $C_1$ is
$$w_0^{728}+364w_0^{260}w_1^{234}w_2^{234}+364w_0^{224}w_1^{252}w_2^{252}.$$

(iii) Let $m=6$, $k=1$, $p=3$. This corresponds to the case $ \upsilon_2(m)= \upsilon_2(l)+1$ and $m\neq2l$. Magma computes that the complete weight enumerator for the code $C_1$ is
$$w_0^{728}+182w_0^{296}w_1^{216}w_2^{216}+546w_0^{224}w_1^{252}w_2^{252}.$$

(iv) Let $m=4$, $k=1$, $p=3$. This corresponds to the case $ \upsilon_2(m)> \upsilon_2(l)+1$. With the help of Magma, we know that the complete weight enumerator for the code $C_1$ is
$$w_0^{80}+60w_0^{32}w_1^{24}w_2^{24}+20w_0^{8}w_1^{36}w_2^{36}.$$

(v) Let $m=2$, $k=1$, $p=3$. This corresponds to the case $m=2l$. Magma works out that the complete weight enumerator for the code $C_1$ is
$$w_0^{8}+2w_1^{4}w_2^{4}.$$

These experimental results coincide with the complete weight enumerators in Theorem \ref{thm:code 2}.
\end{example}

\subsection{The complete weight enumerator for the
code $C_2$}\label{sec:C3}
Recall that
\begin{equation*}
    C_2=\{\mathsf{c}_3(a,b)=(\mathrm{Tr}^m_1(ax^{p^l+1}+bx^2))_{x\in\mathbb{F}_{p^m}^*}:a,b\in\mathbb{F}_{p^m}\}.
\end{equation*}

Now we present the complete weight enumerator of the code $ C_2$ by employing the exponential sum
\begin{equation*}
T(a,b)=\sum_{x\in \mathbb{F}_{p^m}}\zeta_p^{\mathrm{Tr}^m_1(ax^{p^l+1}+bx^2)}.
\end{equation*}

\begin{theorem}\label{thm:code 3} With notation given
before. Let $|R_i|$ and $|K_i|$ be given by
\eqref{def:Ri} and \eqref{def:Ki}, respectively.\\
(i) Assume that $m\neq2l$. Then $C_2$ is a $[p^m-1, 2m]$ cyclic code over $\mathbb{F}_{p}$ and its complete weight enumerator is shown as follows:\\
\textcircled{1} For the case of $s$ and $d$ both being odd, we have
\begin{eqnarray*}\label{1:code 3}
    \mathrm{CWE}(C_2)&=&w_0^{p^m-1}+|R_1|w_0^{p^{m-1}-1}\prod_{\rho\in\mathbb{F}_{p}^*}w_{\rho}^{p^{m-1}+\bar{\eta}(\rho)p^{\frac{m-1}{2}}}\\
                   &&+|R_1|w_0^{p^{m-1}-1}\prod_{\rho\in\mathbb{F}_{p}^*}w_{\rho}^{p^{m-1}-\bar{\eta}(\rho)p^{\frac{m-1}{2}}}\\
                   &&+|R_2|w_0^{p^{m-1}-1+(p-1)p^{\frac{m+d-2}{2}}}\prod_{\rho\in\mathbb{F}_{p}^*}w_{\rho}^{p^{m-1}-p^{\frac{m+d-2}{2}}}\\
                   &&+|R_3|w_0^{p^{m-1}-1-(p-1)p^{\frac{m+d-2}{2}}}\prod_{\rho\in\mathbb{F}_{p}^*}w_{\rho}^{p^{m-1}+p^{\frac{m+d-2}{2}}}\\
                   &&+|R_4|w_0^{p^{m-1}-1}\prod_{\rho\in\mathbb{F}_{p}^*}w_{\rho}^{p^{m-1}+\bar{\eta}(\rho)p^{\frac{m+2d-1}{2}}}\\
                   &&+|R_4|w_0^{p^{m-1}-1}\prod_{\rho\in\mathbb{F}_{p}^*}w_{\rho}^{p^{m-1}-\bar{\eta}(\rho)p^{\frac{m+2d-1}{2}}}\\
\end{eqnarray*}
\textcircled{2} For the case of $s$ being odd and $d$ being even, we have
\begin{eqnarray*}\label{2:code 3}
    \mathrm{CWE}(C_2)&=&w_0^{p^m-1}+|R_1|w_0^{p^{m-1}-1+(p-1)p^{\frac{m-2}{2}}}\prod_{\rho\in\mathbb{F}_{p}^*}w_{\rho}^{p^{m-1}-p^{\frac{m-2}{2}}}\\
                   &&+|R_1|w_0^{p^{m-1}-1-(p-1)p^{\frac{m-2}{2}}}\prod_{\rho\in\mathbb{F}_{p}^*}w_{\rho}^{p^{m-1}+p^{\frac{m-2}{2}}}\\
                   &&+|R_2|w_0^{p^{m-1}-1+(p-1)p^{\frac{m+d-2}{2}}}\prod_{\rho\in\mathbb{F}_{p}^*}w_{\rho}^{p^{m-1}-p^{\frac{m+d-2}{2}}}\\
                   &&+|R_3|w_0^{p^{m-1}-1-(p-1)p^{\frac{m+d-2}{2}}}\prod_{\rho\in\mathbb{F}_{p}^*}w_{\rho}^{p^{m-1}+p^{\frac{m+d-2}{2}}}\\
                   &&+|R_4|w_0^{p^{m-1}-1+(p-1)p^{\frac{m+2d-2}{2}}}\prod_{\rho\in\mathbb{F}_{p}^*}w_{\rho}^{p^{m-1}-p^{\frac{m+2d-2}{2}}}\\
                   &&+|R_4|w_0^{p^{m-1}-1-(p-1)p^{\frac{m+2d-2}{2}}}\prod_{\rho\in\mathbb{F}_{p}^*}w_{\rho}^{p^{m-1}+p^{\frac{m+2d-2}{2}}}\\
\end{eqnarray*}
\textcircled{3} For the case of $s$ being even and $d$ being odd, we have
\begin{eqnarray*}\label{3:code 3}
    \mathrm{CWE}(C_2)&=&w_0^{p^m-1}+|K_1|w_0^{p^{m-1}-1+(p-1)p^{\frac{m-2}{2}}}\prod_{\rho\in\mathbb{F}_{p}^*}w_{\rho}^{p^{m-1}-p^{\frac{m-2}{2}}}\\
                   &&+|K_2|w_0^{p^{m-1}-1-(p-1)p^{\frac{m-2}{2}}}\prod_{\rho\in\mathbb{F}_{p}^*}w_{\rho}^{p^{m-1}+p^{\frac{m-2}{2}}}\\
                   &&+|K_3|w_0^{p^{m-1}-1}\prod_{\rho\in\mathbb{F}_{p}^*}w_{\rho}^{p^{m-1}+\bar{\eta}(\rho)p^{\frac{m+d-1}{2}}}\\
                   &&+|K_3|w_0^{p^{m-1}-1}\prod_{\rho\in\mathbb{F}_{p}^*}w_{\rho}^{p^{m-1}-\bar{\eta}(\rho)p^{\frac{m+d-1}{2}}}\\
                   &&+|K_4|w_0^{p^{m-1}-1+(p-1)p^{\frac{m+2d-2}{2}}}\prod_{\rho\in\mathbb{F}_{p}^*}w_{\rho}^{p^{m-1}-p^{\frac{m+2d-2}{2}}}\\
                   &&+|K_5|w_0^{p^{m-1}-1-(p-1)p^{\frac{m+2d-2}{2}}}\prod_{\rho\in\mathbb{F}_{p}^*}w_{\rho}^{p^{m-1}-p^{\frac{m+2d-2}{2}}}\\
\end{eqnarray*}
\textcircled{4} For the case of $s$ and $d$ both being even, we have
\begin{eqnarray*}\label{4:code 3}
    \mathrm{CWE}(C_2)&=&w_0^{p^m-1}+|K_1|w_0^{p^{m-1}-1+(p-1)p^{\frac{m-2}{2}}}\prod_{\rho\in\mathbb{F}_{p}^*}w_{\rho}^{p^{m-1}-p^{\frac{m-2}{2}}}\\
                   &&+|K_2|w_0^{p^{m-1}-1-(p-1)p^{\frac{m-2}{2}}}\prod_{\rho\in\mathbb{F}_{p}^*}w_{\rho}^{p^{m-1}+p^{\frac{m-2}{2}}}\\
                   &&+|K_3|w_0^{p^{m-1}-1+(p-1)p^{\frac{m+d-2}{2}}}\prod_{\rho\in\mathbb{F}_{p}^*}w_{\rho}^{p^{m-1}-p^{\frac{m+d-2}{2}}}\\
                   &&+|K_3|w_0^{p^{m-1}-1-(p-1)p^{\frac{m+d-2}{2}}}\prod_{\rho\in\mathbb{F}_{p}^*}w_{\rho}^{p^{m-1}+p^{\frac{m+d-2}{2}}}\\
                   &&+|K_4|w_0^{p^{m-1}-1+(p-1)p^{\frac{m+2d-2}{2}}}\prod_{\rho\in\mathbb{F}_{p}^*}w_{\rho}^{p^{m-1}-p^{\frac{m+2d-2}{2}}}\\
                   &&+|K_5|w_0^{p^{m-1}-1-(p-1)p^{\frac{m+2d-2}{2}}}\prod_{\rho\in\mathbb{F}_{p}^*}w_{\rho}^{p^{m-1}+p^{\frac{m+2d-2}{2}}}\\
\end{eqnarray*}
(ii) Assume that $m=2l$. Then $C_2$ is a $[p^m-1, 3m/2]$ cyclic code over $\mathbb{F}_{p}$ and
\textcircled{1} For the case of $d$ being odd, we have
\begin{eqnarray*}\label{51:code 3}
    \mathrm{CWE}(C_2)&=&w_0^{p^m-1}+\frac 12 p^{\frac m2}(p^m-1)w_0^{p^{m-1}+(p-1)p^{\frac{m-2}{2}}-1}\prod_{\rho\in\mathbb{F}_{p}^*}w_{\rho}^{p^{m-1}-p^{\frac{m-2}{2}}}\\
                   &&+\frac 12 p^{\frac m2}(p^{\frac m2}-1)^2w_0^{p^{m-1}-1-(p-1)p^{\frac{m-2}{2}}}\prod_{\rho\in\mathbb{F}_{p}^*}w_{\rho}^{p^{m-1}+p^{\frac{m-2}{2}}}\\
                   &&+\frac 12(p^m-1)w_0^{p^{m-1}-1}\prod_{\rho\in\mathbb{F}_{p}^*}w_{\rho}^{p^{m-1}+\bar{\eta}(\rho)p^{\frac{3m-2}{4}}}\\
                   &&+\frac 12(p^m-1)w_0^{p^{m-1}-1}\prod_{\rho\in\mathbb{F}_{p}^*}w_{\rho}^{p^{m-1}-\bar{\eta}(\rho)p^{\frac{3m-2}{4}}}\\
\end{eqnarray*}
\textcircled{2} For the case of $d$ being even, we have
\begin{eqnarray*}\label{52:code 3}
    \mathrm{CWE}(C_2)&=&w_0^{p^m-1}+\frac 12 p^{\frac m2}(p^m-1)w_0^{p^{m-1}-1+(p-1)p^{\frac{m-2}{2}}}\prod_{\rho\in\mathbb{F}_{p}^*}w_{\rho}^{p^{m-1}-p^{\frac{m-2}{2}}}\\
                   &&+\frac 12 p^{\frac m2}(p^{\frac m2}-1)^2w_0^{p^{m-1}-1-(p-1)p^{\frac{m-2}{2}}}\prod_{\rho\in\mathbb{F}_{p}^*}w_{\rho}^{p^{m-1}+p^{\frac{m-2}{2}}}\\
                   &&+\frac 12(p^m-1)w_0^{p^{m-1}-1+(p-1)p^{\frac{3m-4}{4}}}\prod_{\rho\in\mathbb{F}_{p}^*}w_{\rho}^{p^{m-1}-p^{\frac{3m-4}{4}}}\\
                   &&+\frac 12(p^m-1)w_0^{p^{m-1}-1-(p-1)p^{\frac{3m-4}{4}}}\prod_{\rho\in\mathbb{F}_{p}^*}w_{\rho}^{p^{m-1}+p^{\frac{3m-4}{4}}}\\
\end{eqnarray*}

\end{theorem}
\begin{proof}
(i) Assume that $m\neq2l$. We only give the proof for the case of $s$ and $d$ both being odd,
since other cases are similar to get.

Clearly $a=b=0$ gives the zero codeword and the contribution to the complete
weight enumerator is
$$w_0^{p^m-1}.$$

Consider $a$ or $b$ is nonzero.
Let
\begin{equation*}
    n_{a,b}(\rho)=\sharp\{x\in\mathbb{F}_{p^m}^*:\mathrm{Tr}(ax^{p^l+1}+bx^2)=\rho\},
\end{equation*}
and
\begin{equation*}
    N_{a,b}(\rho)=\sharp\{x\in\mathbb{F}_{p^m}:\mathrm{Tr}(ax^{p^l+1}+bx^2)=\rho\},
\end{equation*} where $\rho\in\mathbb{F}_{p}$.

Recall that $d=\mathrm{gcd}(m,l)$ and $s=m/d$. Since both $m$ and $d$ are odd, we have $m$ and $m-2d$ are odd and $m-d$ is even.

The value of $N_{a,b}(\rho)$
 will be calculated by distinguishing the following subcases.

\emph{Case 1:} $r=m$.

In this case, Lemma \ref{lem:exposumSoddeven} shows that
\begin{eqnarray*}
  T(a,b)=\left\{\begin{array}{lll} \sqrt{(-1)^{\frac{p^d-1}{2}}}p^{\frac m2},&&|R_1|~~\mathrm{times},\\
-\sqrt{(-1)^{\frac{p^d-1}{2}}}p^{\frac m2},&&|R_1|~~\mathrm{times}.\\\end{array} \right.
\end{eqnarray*}

It follows from Equation \eqref{genl:eq} that
\begin{eqnarray*}\label{genl:eqC3}
    N_{a,b}(\rho)&= & p^{m-1}+\frac{1}{p}T(a,b)\sum_{y\in\mathbb{F}_{p}^*}\zeta_p^{y\rho}\bar{\eta}(y^r),\nonumber\\
 &= & p^{m-1}+\frac{1}{p}T(a,b)\sum_{y\in\mathbb{F}_{p}^*}\zeta_p^{y\rho}\bar{\eta}(y),\nonumber\\
 &=& \left\{\begin{array}{lll}p^{m-1},&&~~\mathrm{if~~}\rho=0,\\
p^{m-1}+\frac{1}{p}\bar{\eta}(\rho)T(a,b)G(\bar{\eta},\bar{\chi}),&&~~\mathrm{otherwise}.
\\\end{array} \right.
\end{eqnarray*}
Then
\begin{eqnarray*}
  n_{a,b}(\rho) = \left\{\begin{array}{lll}p^{m-1}-1,&&~~\mathrm{if~~}\rho=0,\\
p^{m-1}+\frac{1}{p}\bar{\eta}(\rho)T(a,b)G(\bar{\eta},\bar{\chi}),&&~~\mathrm{otherwise}.\\\end{array} \right.
\end{eqnarray*}

By Lemma \ref{lm:gauss sum}, the contributions of such terms to the complete weight enumerator
is then
\begin{eqnarray*}
   |R_1|w_0^{p^{m-1}-1}\prod_{\rho\in\mathbb{F}_{p}^*}w_{\rho}^{p^{m-1}+\bar{\eta}(\rho)p^{\frac{m-1}{2}}}
   +|R_1|w_0^{p^{m-1}-1}\prod_{\rho\in\mathbb{F}_{p}^*}w_{\rho}^{p^{m-1}-\bar{\eta}(\rho)p^{\frac{m-1}{2}}}.\\         \end{eqnarray*}

\emph{Case 2:} $r=m-d$.

In this case, we have
\begin{eqnarray*}
  T(a,b)=\left\{\begin{array}{lll} p^{\frac{m+d}2},&&|R_2|~~times,\\
-p^{\frac{m+d}2},&&|R_3|~~times.\\\end{array} \right.
\end{eqnarray*}

Equation \eqref{genl:eq} yields that
\begin{eqnarray*}\label{genl:eqC3}
    N_{a,b}(\rho)&= & p^{m-1}+\frac{1}{p}T(a,b)\sum_{y\in\mathbb{F}_{p}^*}\zeta_p^{y\rho}\\
 &=& \left\{\begin{array}{lll}p^{m-1}+\frac{p-1}{p}T(a,b),&&~~\mathrm{if~~}\rho=0,\\
p^{m-1}-\frac{1}{p}T(a,b),&&~~\mathrm{otherwise}.
\end{array} \right.
\end{eqnarray*}

Therefore
\begin{eqnarray*}
  n_{a,b}(\rho) = \left\{\begin{array}{lll}p^{m-1}-1+\frac{p-1}{p}T(a,b),&&~~\mathrm{if~~}\rho=0,\\
p^{m-1}-\frac{1}{p}T(a,b),&&~~\mathrm{otherwise}.
\end{array} \right.
\end{eqnarray*}

The contributions of such terms to the complete weight enumerator
is
\begin{eqnarray*}
&|R&_2|w_0^{p^{m-1}+(p-1)p^{\frac{m+d-2}{2}}-1}\prod_{\rho\in\mathbb{F}_{p}^*}w_{\rho}^{p^{m-1}-p^{\frac{m+d-2}{2}}}\\
                   &&+|R_3|w_0^{p^{m-1}-(p-1)p^{\frac{m+d-2}{2}}-1}\prod_{\rho\in\mathbb{F}_{p}^*}w_{\rho}^{p^{m-1}+p^{\frac{m+d-2}{2}}}\\
\end{eqnarray*}

\emph{Case 3:} $r=m-2d$.

In this case, we have
\begin{eqnarray*}
  T(a,b)=\left\{\begin{array}{lll} \sqrt{(-1)^{\frac{p^d-1}{2}}}p^{\frac {m+2d}2},&&|R_{4}|~~times,\\
-\sqrt{(-1)^{\frac{p^d-1}{2}}}p^{\frac {m+2d}2},&&|R_{4}|~~times.\\\end{array} \right.
\end{eqnarray*}

By a similar discussion to \emph{Case 1}, we have
\begin{eqnarray*}
  n_{a,b}(\rho) = \left\{\begin{array}{lll}p^{m-1}-1,&&~~\mathrm{if~~}\rho=0,\\
p^{m-1}+\frac{1}{p}\bar{\eta}(\rho)T(a,b)G(\bar{\eta},\bar{\chi}),&&~~\mathrm{otherwise}.\\\end{array} \right.
\end{eqnarray*}

Therefore, the contributions of such terms to the complete weight enumerator
is then

\begin{eqnarray*}
  |R_{4}|w_0^{p^{m-1}-1}\prod_{\rho\in\mathbb{F}_{p}^*}w_{\rho}^{p^{m-1}+\bar{\eta}(\rho)p^{\frac{m+2d-1}{2}}}+|R_{4}|w_0^{p^{m-1}-1}\prod_{\rho\in\mathbb{F}_{p}^*}w_{\rho}^{p^{m-1}-\bar{\eta}(\rho)p^{\frac{m+2d-1}{2}}}.\\
\end{eqnarray*}

The desired conclusion then follows immediately from the above arguments.

(ii) Assume that $ m=2l$.

Let $$K=\{x\in\mathbb{F}_{p^m}\big|~x^{p^l}+x=0\}.$$
 Note that $\mathsf{c}_3(a,b)=\mathsf{c}_3(a+\tau,b)$
 for any $\tau\in K$ and $\mathsf{c}_3(a,b)\in C_2$. Hence, $C_2$ is degenerate with
 dimension $3m/2$ over $\mathbb{F}_p$.

Clearly $|K|=p^{\frac{m}{2}}$ and $s=2$. Substituting $d={m}/{2}$ to
 the case $s$ being even and dividing each frequency by
 $p^{\frac{m}{2}}$, we get the desired results.

This finishes the proof of Theorem \ref{thm:code 3}.
\hfill\space$\qed$\end{proof}

\begin{example}
(i) Let $m=3$, $k=1$, $p=3$. Then $s=3$ and $d=1$. Magma works out that $C_2$ is a [26, 6, 12] cyclic code with the complete weight enumerator
\begin{eqnarray*}
 w_0^{26 }&+& 156w_0^{14}w_1^{6}w_2^{6 }+ 13w_0^{8}w_1^{18 }+ 234w_0^{8}w_1^{12}w_2^{6 }\\
    &+&234w_0^{8}w_1^{6}w_2^{12 }+ 13w_0^{8}w_2^{18 }+ 78w_0^{2}w_1^{12}w_2^{12}.
\end{eqnarray*}

(ii) Let $m=6$, $k=2$, $p=3$. Then $s=3$ and $d=2$. Magma works out that $C_2$ is a [728, 12, 324] cyclic code with the complete weight enumerator
\begin{eqnarray*}
w_0^{728 }&+ &364w_0^{404}w_1^{162}w_2^{162 }+
            32760w_0^{296}w_1^{216}w_2^{216 }+235872w_0^{260}w_1^{234}w_2^{234 }\\
    &+& 235872w_0^{224}w_1^{252}w_2^{252 }+26208w_0^{188}w_1^{270}w_2^{270 }+ 364w_0^{80}w_1^{324}w_2^{324}.
\end{eqnarray*}

(iii) Let $m=4$, $k=2$, $p=3$. Then $s=2$ and $d=2$. Magma works out that $C_2$ is a [80, 6, 36] cyclic code with the complete weight enumerator
\begin{eqnarray*}
w_0^{80 }+ 40w_0^{44}w_1^{18}w_2^{18 }+ 360w_0^{32}w_1^{24}w_2^{24 }+ 288w_0^{20}w_1^{30}w_2^{30 }+
    40w_0^{8}w_1^{36}w_2^{36}.
\end{eqnarray*}

These experimental results coincide with the complete weight enumerators in Theorem \ref{thm:code 3}.

\end{example}

\section{Conclusion and remarks}\label{sec:conclusion}

In this paper, we concentrated on the complete weight enumerators of
cyclic codes. A general strategy was proposed by using exponential sums and then the complete weight enumerators of three classes of cyclic codes were explicitly determined. In addition, one can get the weight distributions of the codes through their complete weight enumerators.

It should be noted that the exponential sums are known in a few cases. Hence the
complete weight enumerator of most cyclic codes cannot be explicitly
presented. We mention that the complete weight enumerators are still open for most cyclic codes and it will be a good research problem to construct more cyclic codes and determine their complete weight enumerators and weight distributions as well. We leave this for future work.

\begin{acknowledgements}
The work of Zheng-An Yao is partially supported by the NSFC (Grant No.11271381), the NSFC (Grant No.11431015)
and China 973 Program (Grant No. 2011CB808000).
This work is also partially supported by the NSFC (Grant No. 61472457) and Guangdong Natural Science
Foundation (Grant No. 2014A030313161).
\end{acknowledgements}


\begin{thebibliography}{10}
\providecommand{\url}[1]{{#1}}
\providecommand{\urlprefix}{URL }
\expandafter\ifx\csname urlstyle\endcsname\relax
  \providecommand{\doi}[1]{DOI~\discretionary{}{}{}#1}\else
  \providecommand{\doi}{DOI~\discretionary{}{}{}\begingroup
  \urlstyle{rm}\Url}\fi

\bibitem{BaeLi2015complete}
Bae, S., Li, C., Yue, Q.: On the complete weight enumerators of some reducible
  cyclic codes.
\newblock Discrete Mathematics \textbf{338}(12), 2275 -- 2287 (2015).


\bibitem{Blake1991}
Blake, I.F., Kith, K.: {On the complete weight enumerator of Reed-Solomon
  codes}.
\newblock SIAM J. Discret. Math. \textbf{4}(2), 164--171 (1991)

\bibitem{chu2006constant}
Chu, W., Colbourn, C.J., Dukes, P.: On constant composition codes.
\newblock Discrete Applied Mathematics \textbf{154}(6), 912--929 (2006)

\bibitem{coulter1998explicit}
Coulter, R.S.: Explicit evaluations of some {W}eil sums.
\newblock Acta Arithmetica \textbf{83}(3), 241--251 (1998)

\bibitem{delsarte1975subfield}
Delsarte, P.: On subfield subcodes of modified {R}eed-{S}olomon codes.
\newblock IEEE Transactions on Information Theory \textbf{21}(5), 575--576
  (1975)

\bibitem{ding2008optimal}
Ding, C.: Optimal constant composition codes from zero-difference balanced
  functions.
\newblock IEEE Transactions on Information Theory \textbf{54}(12), 5766--5770
  (2008)

\bibitem{ding2007generic}
Ding, C., Helleseth, T., Klove, T., Wang, X.: {A generic construction of
  Cartesian authentication codes}.
\newblock IEEE Transactions on Information Theory \textbf{53}(6), 2229--2235
  (2007)

\bibitem{ding2011}
Ding, C., Liu, Y., Ma, C., Zeng, L.: The weight distributions of the duals of
  cyclic codes with two zeros.
\newblock IEEE Transactions on Information Theory \textbf{57}(12), 8000--8006
  (2011)

\bibitem{Ding2005auth}
Ding, C., Wang, X.: A coding theory construction of new systematic
  authentication codes.
\newblock Theoretical computer science \textbf{330}(1), 81--99 (2005)

\bibitem{ding2013hamming}
Ding, C., Yang, J.: Hamming weights in irreducible cyclic codes.
\newblock Discrete Mathematics \textbf{313}(4), 434--446 (2013)

\bibitem{ding2006construction}
Ding, C., Yin, J.: A construction of optimal constant composition codes.
\newblock Designs, Codes and Cryptography \textbf{40}(2), 157--165 (2006)

\bibitem{dinh2015recent}
Dinh, H.Q., Li, C., Yue, Q.: Recent progress on weight distributions of cyclic
  codes over finite fields.
\newblock Journal of Algebra Combinatorics Discrete Structures and Applications
  \textbf{2}(1), 39--63 (2015)

\bibitem{draper2007explicit}
Draper, S., Hou, X.: Explicit evaluation of certain exponential sums of
  quadratic functions over $\mathbb{F}_{p^n}$, $p$ odd.
\newblock http://arxiv.org/pdf/0708.3619v1.pdf  (2007)

\bibitem{helleseth2006}
Helleseth, T., Kholosha, A.: Monomial and quadratic bent functions over the
  finite fields of odd characteristic.
\newblock IEEE Transactions on Information Theory \textbf{52}(5), 2018--2032
  (2006)

\bibitem{kith1989complete}
Kith, K.: {Complete weight enumeration of Reed-Solomon codes}.
\newblock Master's thesis, Department of Electrical and Computing Engineering,
  University of Waterloo, Waterloo, Ontario, Canada  (1989)

\bibitem{kuzmin1999complete}
Kuzmin, A., Nechaev, A.: {Complete weight enumerators of generalized Kerdock
  code and linear recursive codes over Galois ring}.
\newblock In: Workshop on coding and cryptography, pp. 333--336 (1999)

\bibitem{kuzmin2001complete}
Kuzmin, A., Nechaev, A.: {Complete weight enumerators of generalized Kerdock
  code and related linear codes over Galois ring}.
\newblock Discrete applied mathematics \textbf{111}(1), 117--137 (2001)

\bibitem{li2015complete}
Li, C., Yue, Q., Fu, F.W.: Complete weight enumerators of some cyclic codes.
\newblock Designs, Codes and Cryptography, (2015).
\newblock Doi:10.1007/s10623-015-0091-5.

\bibitem{lidl1983finite}
Lidl, R., Niederreiter, H.: Finite fields.
\newblock Encyclopedia of Mathematics and its Applications. Reading,
  Massachusetts, USA: Addison-Wesley \textbf{20} (1983)

\bibitem{luo2008weight}
Luo, J., Feng, K.: On the weight distributions of two classes of cyclic codes.
\newblock IEEE Transactions on Information Theory \textbf{54}(12), 5332--5344
  (2008)

\bibitem{macwilliams1977theory}
MacWilliams, F.J., Sloane, N.J.A.: The theory of error-correcting codes,
  vol.~16.
\newblock North-Holland Publishing, Amsterdam (1977)

\bibitem{sharma2012weight}
Sharma, A., Bakshi, G.K.: The weight distribution of some irreducible cyclic
  codes.
\newblock Finite Fields and Their Applications \textbf{18}(1), 144--159 (2012)

\bibitem{vega2012weight}
Vega, G.: The weight distribution of an extended class of reducible cyclic
  codes.
\newblock IEEE Transactions on Information Theory \textbf{58}(7), 4862--4869
  (2012)

\bibitem{yu2014weight}
Yu, L., Liu, H.: The weight distribution of a family of $p$-ary cyclic codes.
\newblock Designs, Codes and Cryptography  (2014)
\newblock Doi:10.1007/s10623-014-0029-3.

\bibitem{zheng2013weight}
Zheng, D., Wang, X., Zeng, X., Hu, L.: The weight distribution of a family of
  $p$-ary cyclic codes.
\newblock Designs, Codes and Cryptography \textbf{75}(2),263--275 (2015).

\end{thebibliography}

\end{document}